\documentclass[12pt]{article}

\usepackage{geometry}
\usepackage{enumitem}
\usepackage{cite}
\usepackage{graphicx}
\usepackage{subfigure}
\usepackage{amssymb}
\usepackage{amsmath}
\allowdisplaybreaks[1]

\usepackage{amsthm}

\newtheorem{theorem}{Theorem}
\newtheorem{lemma}[theorem]{Lemma}

\makeatletter
\def\@endtheorem{\endtrivlist}
\makeatother

\newcounter{brule}
\newenvironment{brule}{\refstepcounter{brule}\par\smallskip\noindent
\textbf{(B\arabic{brule})}\quad}{}
\newcommand{\currentrule}{B\arabic{brule}}

\newcommand{\Fmin}[1]{\mathcal{F}(#1)}

\begin{document}

\title{Faster parameterized algorithm for Bicluter Editing}
\author{Dekel Tsur%
\thanks{Ben-Gurion University of the Negev.
Email: \texttt{dekelts@cs.bgu.ac.il}}}
\date{}
\maketitle

\begin{abstract}
In the \textsc{Bicluter Editing} problem the input is a graph $G$ and an integer
$k$, and the goal is to decide whether $G$ can be transformed into a bicluster
graph by adding and removing at most $k$ edges.
In this paper we give an algorithm for \textsc{Bicluster Editing} whose running time is $O^*(3.116^k)$.
\end{abstract}


\section{Introduction}
A graph $G$ is called a \emph{biclique} if $G$ is a complete bipartite graph,
namely, there is a partition of $V(G)$ into disjoint non-empty sets $V_1,V_2$
such that $V_1$ and $V_2$ are independent sets and there is an edge between
every vertex in $V_1$ and every vertex in $V_2$.
A graph $G$ is called a \emph{bicluster graph} if every connected component
of $G$ with at least two vertices is a biclique.
In the \textsc{Bicluter Editing} problem the input is a graph $G$ and an integer
$k$, and the goal is to decide whether $G$ can be transformed into a bicluster
graph by adding and removing at most $k$ edges.

Protti et al.~\cite{protti2006applying} gave an $O^*(4^k)$-time algorithm for
\textsc{Bicluter Editing}.
A faster algorithm, with $O^*(3.237^k)$ running time, was given by
Guo et al.~\cite{guo2008improved}.
In this paper, we give an algorithm for \textsc{Bicluter Editing} with
$O^*(3.116^k)$ running time.

\paragraph{Preleminaries}
For a set $S$ of vertices in a graph $G$, $G[S]$ is the subgraph
of $G$ induced by $S$ (namely, $G[S]=(S,E_S)$ where
$E_S = \{(u,v) \in E(G) \colon u,v \in S\}$).
%For a graph $G = (V,E)$ and a set $F \subseteq E$ of edges,
%$G-F$ is the graph $(V,E \setminus F)$.
For a graph $G = (V,E)$ and a set $F$ of pairs of vertices,
$G \triangle F$ is the graph $(V, (E \setminus F) \cup (F \setminus E))$.
A set $F$ of pairs of vertices is called an \emph{editing set} of a graph $G$
if $G \triangle F$ is a bicluster graph.
For a graph $G$, let $\Fmin{G}$ be a set containing every inclusion minimal
editing set of $G$.

A $P_4$ is a graph which consists of a path on 4 vertices.
Let $A$ be a set of vertices that induces a $P_4$ in a graph $G$.
Let
$I(A) = \{ v\in V(G) \setminus A \colon N(v)\cap A = \emptyset\}$
and 
$P(A) = V(G) \setminus (A\cup I(A)) =
\{ v\in V(G) \setminus A \colon N(v)\cap A \neq \emptyset\}$.

\section{The algorithm}

A graph is a bicluster graph if and only if it is bipartite and
it does not contain an induced $P_4$.
Therefore, the \textsc{Bicluster Editing} problem is closely related
to the \textsc{Cograph Editing} problem, which is the problem of deciding
whether a graph $G$ can be transformed to a graph without an induced $P_4$
by adding and removing at most $k$ edges.
Our algorithm for \textsc{Bicluster Editing} is based on the algorithm
of~\cite{Tsur_cograph} for \textsc{Cograph Editing}.

The algorithm is a branching algorithm (cf.~\cite{cygan2015parameterized}).
The algorithm uses the following branching rules.
\begin{brule}
Let $X$ be a set that induces a triangle.
For every edge $e$ in $G[X]$,
recurse on the instance $(G \triangle \{e\},k-1)$.\label{rule:triangle}
\end{brule}

The branching factor of Rule~(\currentrule) is $(1,1,1)$ and
the branching number is 3.

\begin{brule}
Let $A$ be a set that induces a $P_4$ such that $|P(A)| \geq 2$.
Choose distinct vertices $p,p' \in P(A)$.
For every $F \in \Fmin{G[A\cup \{p,p'\}]}$,
recurse on the instance $(G \triangle F,k-|F|)$.\label{rule:P}
\end{brule}

To compute the branching number of Rule~(\currentrule), we consider all
possible cases for the neighbors of $p$ in $A$ and all possible cases
for the neighbors of $p'$ in $A \cup \{p\}$. Note that the number of cases
is finite.
For each case, we used a Python script to compute $\Fmin{G[A\cup \{p,p'\}]}$
and to compute the corresponding branching number.
The case with the largest branching number is when $p$ is adjacent to the
first vertex of the path and $p'$ is adjacent to the second vertex of the
path. Additionally, $p$ and $p'$ are not adjacent.
In this case, the branching vector is $(2,2,2,2,2,2,2,2,3,3,3,3,3,4)$
and the branching number is at most 3.116.

\begin{brule}
Let $A$ be a set that induces a $P_4$ such that there is
a vertex $p \in P(A)$ that is adjacent to a vertex $i \in I(A)$.
For every $F \in \Fmin{G[A\cup \{p,i\}]}$,
recurse on the instance $(G \triangle F,k-|F|)$.\label{rule:PI}
\end{brule}

The branching number of Rule~(\currentrule) was also computed with
a script. The branching number of this rule is at most 3.116.

We now consider an instance of the problem on which the above
branching rules cannot be applied.
We will show that the instance can be solved in polynomial time.

\begin{lemma}
Let $(G,k)$ be an instance of \textsc{Bicluster Editing} on which
Rules~(B\ref{rule:triangle})--(B\ref{rule:PI}) cannot be applied.
Then, every connected component of $G$ with at least 6 vertices
is a biclique.\label{lem:reduced}
\end{lemma}
\begin{proof}
Without loss of generality, we assume that $G$ is connected.
We also assume that $G$ has at least 6 vertices.
We first claim that $\overline{G}$ is not connected.
Suppose conversely that $\overline{G}$ is connected.
By a result of Seinsche~\cite{seinsche1974property},
$G$ contains an induced $P_4$,
and let $A$ be a set of vertices that induces a $P_4$ in $G$.
Since $G$ has at least 6 vertices and $|P(A)| \leq 1$
(due to Rule~(B\ref{rule:P})), we have that $I(A) \neq \emptyset$.
By definition, a vertex $i \in I(A)$ is not adjacent to the vertices in $A$.
Additionally, if $P(A) \neq \emptyset$, a vertex $i \in I(A)$ is not adjacent to
the single vertex in $P(A)$ (due to Rule~(B\ref{rule:PI})).
Therefore, $A \cup P(A)$ is a connected component in $G$ which does not contain
all the vertices of $G$, contradicting the assumption that $G$ is connected.
Therefore, $\overline{G}$ is not connected.

Let $C_1,\ldots,C_p$ be the connected components of $\overline{G}$.
In $G$, a vertex in a connected component $C_i$ is adjacent to all the vertices
in $V(G)\setminus C_i$.
Since $G$ does not contain a triangle (due to Rule~(B\ref{rule:triangle})),
it follows that $p = 2$.
Additionally, the sets $C_1$ and $C_2$ are independent sets.
Therefore, $G$ is a biclique.
\end{proof}

\begin{lemma}
Let $G$ be a graph. If $F$ is a minimum size editing set of $G$ then
for every $(u,v)\in F$, $u$ and $v$ belong to the same connected component of
$G$.\label{lem:cc}
\end{lemma}
\begin{proof}
Let $F$ be a minimum size editing set of $G$.
Let $F'$ be a set containing every pair $(u,v) \in F$ such that $u,v$ belong
to the same connected component of $G$.
Since being a bicluster graph is a hereditary property, for every
connected component $C$ of $G$ we have that
$(G\triangle F)[C] = (G\triangle F')[C]$ is a bicluster graph.
%Therefore, in the graph $G\triangle F'$, every connected component $C$ of $G$
%is transformed into a bicluster graph.
Since the disjoint union of bicluster graphs is a bicluster graph,
$G\triangle F'$ is a bicluster graph.
Since $F$ is a minimum size editing set of $G$, it follows that $F' = F$.
\end{proof}

By Lemma~\ref{lem:reduced} and Lemma~\ref{lem:cc},
if $(G,k)$ is an instance on which the above branching rules cannot be applied
then the instance can be solved in polynomial time:
Let $C_1,\ldots,C_p$ be the connected components of $G$.
For every $i$ such that $|C_i| \leq 5$, compute a minimum size editing set
$F_i$ of $G[C_i]$ using brute force.
For every $i$ such that $|C_i| \geq 6$, let $F_i = \emptyset$.
Then, $F = \bigcup_{i=1}^p F_i$ is a minimum size editing set of $G$.

We obtain the following theorem.
\begin{theorem}
There is an $O^*(3.116^k)$-time algorithm for \textsc{Bicluster Editing}.
\end{theorem}

\bibliographystyle{abbrv}
\bibliography{bicluster,dekel}

\end{document}